\documentclass[letterpaper, 10 pt, conference]{ieeeconf} 
\IEEEoverridecommandlockouts  

\usepackage{amsmath}
\usepackage{amssymb}
\usepackage{amsthm}
\usepackage{mathtools}
\usepackage{derivative}
\usepackage{xcolor}
\usepackage{bm}

\usepackage[noadjust]{cite}
\usepackage{url}

\usepackage{graphicx}
\usepackage[export]{adjustbox} 
\graphicspath{{figures/}}

\newtheorem{theorem}{Theorem}
\newtheorem{lemma}{Lemma}

\theoremstyle{definition}
\newtheorem{definition}{Definition}

\newtheorem{remark}{Remark}

\newtheorem{example}{Example}

\newcommand{\argmin}{\operatornamewithlimits{arg\,min}}

\newcommand{\R}{\mathbb{R}}

\newcommand{\mc}[1]{\mathcal{#1}}

\newcommand{\act}{\operatorname{Act}}

\newcommand{\bzero}{\mathbf{0}}

\newcommand{\bd}{\mathbf{d}}
\newcommand{\be}{\mathbf{e}}
\renewcommand{\bf}{\mathbf{f}} 
\newcommand{\bg}{\mathbf{g}}

\newcommand{\bk}{\mathbf{k}}

\newcommand{\bp}{\mathbf{p}}
\newcommand{\bq}{\mathbf{q}}

\newcommand{\bu}{\mathbf{u}}
\newcommand{\bv}{\mathbf{v}}

\newcommand{\bx}{\mathbf{x}}
\newcommand{\by}{\mathbf{y}}
\newcommand{\bz}{\mathbf{z}}

\newcommand{\bF}{\mathbf{F}}

\newcommand{\bR}{\mathbf{R}}

\newcommand{\bkappa}{\bm{\kappa}}

\newcommand{\bpi}{\bm{\pi}}
\newcommand{\bpsi}{\bm{\psi}}

\newcommand{\bomega}{\bm{\omega}}

\usepackage{xcolor}
\definecolor{myblue}{RGB}{49, 114, 174}
\definecolor{myred}{rgb}{0.796, 0.235, 0.2}
\definecolor{mygreen}{rgb}{0.22, 0.596, 0.149}
\definecolor{mypurple}{rgb}{0.584,0.345,0.698}

\usepackage{blindtext}

\title{\textbf{Safety-Critical Controller Synthesis with Reduced-Order Models}}

\author{Max H. Cohen$^1$, Noel Csomay-Shanklin$^1$, William D. Compton$^1$, Tamas G. Molnar$^2$, and Aaron D. Ames$^1$ %
\thanks{$^1$The authors are with the Department of Mechanical and Civil Engineering, California Institute of Technology, Pasadena, CA \texttt{\{maxcohen,noelcs,wcompton,ames\}@caltech.edu}.}
\thanks{$^2$ The author is with the Department of Mechanical Engineering, Wichita State University, Wichita, KS \texttt{\{tamas.molnar\}@wichita.edu}.}
\thanks{This research was supported by Boeing and NSF CPS Award \#1932091.}
}

\begin{document}
\maketitle
\begin{abstract}
    Reduced-order models (ROMs) provide lower dimensional representations of complex systems, capturing their salient features while simplifying control design. Building on previous work, this paper presents an overarching framework for the integration of ROMs and control barrier functions, enabling the use of simplified models to construct safety-critical controllers while providing safety guarantees for complex full-order models. To achieve this, we formalize the connection between full and ROMs by defining projection mappings that relate the states and inputs of these models and leverage simulation functions to establish conditions under which safety guarantees may be transferred from a ROM to its corresponding full-order model. The efficacy of our framework is illustrated through simulation results on a drone and hardware demonstrations on ARCHER, a 3D hopping robot.
\end{abstract}

\section{Introduction}
Control barrier functions (CBFs) \cite{AmesTAC17} have proven successful in designing safety-critical controllers for nonlinear systems. Despite this, developing general procedures for constructing CBFs for high-dimensional complex systems has remained elusive \cite{CohenARC24}. More recently, the authors have attempted to leverage reduced-order models (ROMs) 
to construct CBFs for simple models that may be refined to ensure the safety of more complex, full-order systems \cite{TamasRAL22,TamasACC23,CohenARC24}. This paradigm may be traced back to \cite{TamasRAL22} wherein simple kinematic models were used to generate safe velocity commands to be tracked by more complicated robotic systems. Such ideas were expanded on in \cite{TamasACC23} to address ROMs with bounded inputs and in \cite{CohenARC24} where this ROM paradigm is related to nonlinear control techniques such as backstepping \cite{Krstic}. In this paper, we unify and generalize previous developments \cite{TamasRAL22,TamasACC23,CohenARC24} to provide a formal framework for leveraging ROMs in the context of safety-critical control with CBFs.

The paradigm of safety-critical control with CBFs and ROMs is closely related to planner-tracker frameworks in which reduced-order planning models are used to generate trajectories that are tracked by full-order tracking models. The majority of planner-tracker methods focus on the construction of a tracking controller and associated tracking error bound using methods such as Hamilton-Jacobi reachability \cite{TomlinTAC21}, sum of squares programming \cite{ArcakARC22}, model predictive control \cite{BendersArXiv24}, or contraction theory \cite{SinghIJRR23}. 
Within the context of planner-tracker frameworks, we address a problem converse to those above: rather than designing a tracking controller for a given planning model, we focus on designing safe reduced-order reference commands for a full-order system with a fixed tracking controller. 
This is motivated by the observation that many systems of interest are equipped with high-performance, but ``black-box," tracking controllers that are not easily modified.
The canonical example is robotic systems, where one can often send ``joystick'' commands, without knowledge of the underlying tracking controller. 

\begin{figure}
    \centering
    \includegraphics{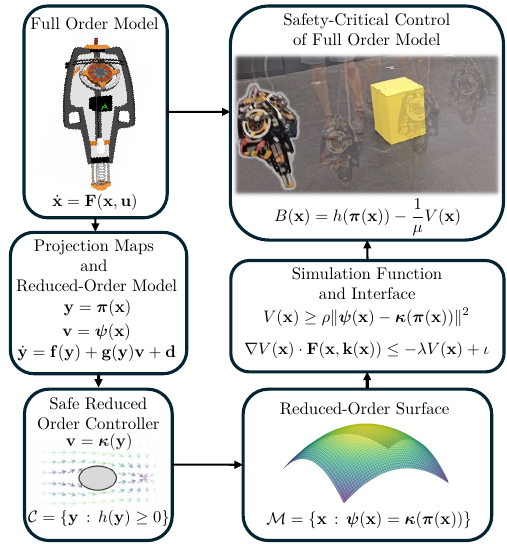}
    \vspace{-0.8cm}
    \caption{\textbf{Overview:} We project high-dimensional control systems onto reduced-order spaces, design safety-critical controllers for a reduced-order representation of the original system, and then relate the inputs of this reduced-order system back to the full-order system. Our theoretical developments are illustrated through their application to ARCHER, a 3D hopping robot, a video of which is available at \texttt{https://vimeo.com/1010060590?share=copy}.}
    \label{fig:hero-figure}
\end{figure}

The perspective taken herein also has roots in classical works from the hybrid systems community on abstraction-based control \cite{Tabuada,Belta}. In particular, our ROM paradigm is inspired by the notion of $\phi$-related control systems \cite{PappasTAC00,PappasTAC02} and we leverage approximate simulation relations \cite{GirardAutomatica09} to formalize the connection between full and reduced-order models. While these works offer a powerful framework to formally reason about system abstractions, their application has often been limited to simple, low-dimensional systems. Here, we expand the applicability of such ideas to more complex systems by illustrating how concepts such as simulation relations integrate with CBFs, leading to practical constructions of safety-critical controllers.

In this paper, we present a unified framework for ROMs in the context of safety-critical control with CBFs, generalizing the ideas in \cite{TamasRAL22,TamasACC23,CohenARC24}. Specifically, we introduce projection mappings that relate the states and inputs of full-order and reduced-order models, which allow for transferring properties of a ROM back to its corresponding full-order model. A key assumption enabling our approach is the existence of a simulation function and corresponding interface \cite{GirardAutomatica09} that refines reduced-order inputs to those for the full-order model.
At a high level, this assumption formalizes the capability of a full-order model to track commands generated by a suitable ROM.
While this may seem restrictive, we argue that for many practical systems, it is not: drones and other unmanned aerial vehicles often come equipped with well-designed tracking controllers \cite{LeeCDC10,MellingerICRA2011}, whereas state-of-the-art methods in robotic locomotion rely on converting velocity references into joint torques using model predictive control \cite{WensingTRO24} or reinforcement learning \cite{HutterScience19}. Rather than redesigning the control architecture of such systems to incorporate safety considerations, the perspective taken herein is to leverage these existing architectures by passing them suitably designed reduced-order inputs, which, as we will demonstrate experimentally, leads to practical safety guarantees.

Our approach generalizes previous results on CBFs and ROMs \cite{TamasRAL22,TamasACC23,CohenARC24} as follows. We consider a larger class of FOMs and ROMs than those in \cite{TamasRAL22,CohenARC24}, which focused on fully actuated robotic systems and strict-feedback systems, respectively. We also provide a Lyapunov-like characterization of tracking controllers via simulation functions \cite{GirardAutomatica09}, which, compared to \cite{TamasACC23}, yields time-invariant safe sets rather than time-varying safe sets and relaxes other conditions such as only using CBFs with bounded gradients. Furthermore, we showcase the efficacy of our framework on ARCHER \cite{Ambrose,NoelICRA23,NoelIROS24}, a highly underactuated 3D hopping robot, advancing the practical applications of ROMs within a CBF framework. 

\section{Preliminaries and Problem Formulation}\label{sec:prelim}
\noindent\textbf{Notation:} A continuous function $\alpha\,:\,\R_{\geq0}\rightarrow\R_{\geq0}$ is said to be a class $\mc{K}$ function ($\alpha\in\mc{K}$) if $\alpha(0)=0$ and $\alpha$ is strictly increasing. For a set $\mc{S}$ we use $\partial\mc{S}$ to denote its boundary. A real number $a\in\R$ is said to be a regular value of a scalar function $h\,:\,\R^n\rightarrow\R$ if $h(\bx)=a$ implies $\nabla h(\bx)\neq\bzero$.
\vspace{0.1cm}

\noindent\textbf{Safety:} Consider a system with state $\bx\in\R^n$ and dynamics:
\begin{equation}\label{eq:dyn}
    \dot{\bx} = \bf(\bx),
\end{equation}
where $\bf\,:\,\R^n\rightarrow\R^n$ is a locally Lipschitz vector field. Under this Lipschitz assumption, for each initial condition $\bx_0\in\R^n$ the dynamics \eqref{eq:dyn} generate a unique continuously differentiable trajectory $\bx\,:\, I(\bx_0)\rightarrow\R^n$ satisfying \eqref{eq:dyn} on some maximal interval of existence $I(\bx_0)\subseteq\R_{\geq0}$. A set $\mc{S}\subset\R^n$ is said to be forward invariant for \eqref{eq:dyn} if for each initial condition $\bx_0\in\mc{S}$, the resulting trajectory $t\mapsto\bx(t)$ satisfies $\bx(t)\in\mc{S}$ for all $t\in I(\bx_0)$. The following result, known as Nagumo's Theorem, provides necessary and sufficient conditions for set invariance.

\begin{theorem}\label{theorem:nagumo}
    A closed set $\mc{S}\subset\R^n$ is forward invariant for \eqref{eq:dyn} if and only if $\bf(\bx)\in\mc{T}_{\mc{S}}(\bx)$ for all $\bx\in\partial\mc{S}$, where $\mc{T}_{\mc{S}}(\bx)$ denotes the contingent cone\footnote{See \cite[Ch. 4]{Blanchini} for a precise definition of the contingent cone.} to $\mc{S}$ at $\bx$.
\end{theorem}

Informally, Nagumo's Theorem states that a set is forward invariant if and only if, for each $\bx\in\partial\mc{S}$, the vector field characterizing the dynamics points into $\mc{S}$. Further details of Nagumo's Theorem can be found in \cite[Ch. 4]{Blanchini} and \cite[Ch. 4]{AbrahamMarsdenRatiu}. In this paper, we associate the concept of set invariance with that of safety: a system is safe if it remains in a desirable subset of the state space. In what follows, our main objective is to design safety-critical controllers for complex high-dimensional systems using relatively simple, or, reduced-order, models within the framework of CBFs \cite{CohenARC24}.

\section{Reduced-Order Models}\label{sec:rom}
Reduced-order models (ROMs) provide lower-dimensional representations of systems, capturing the salient features of more complex models while simplifying controller design. To formalize this idea, consider a control system:
\begin{equation}\label{eq:full-order-dyn}
    \dot{\bx} = \bF(\bx,\bu),
\end{equation}
with state $\bx\in\R^N$, input $\bu\in\R^M$ and locally Lipschitz dynamics $\bF\,:\,\R^N\times\R^M\rightarrow\R^N$,
which we refer to as the full-order model (FOM). To obtain a reduced-order representation of \eqref{eq:full-order-dyn}, we define 
a differentiable \emph{state projection map} $\bpi\,:\,\R^N\rightarrow\R^n$ and a \emph{control projection map} $\bpsi\,:\,\R^N\rightarrow\R^m$, which map the full-order state $\bx\in\R^N$ to a reduced-order state $\by\in\R^n$ and input $\bv\in\R^m$ as:
\begin{equation}
    \by \coloneqq \bm{\pi}(\bx),\quad \bv \coloneqq \bm{\psi}(\bx).
\end{equation}
This state projection map $\bpi$ allows for defining the dynamics of the FOM \eqref{eq:full-order-dyn} projected onto the reduced-order space:
\begin{equation}\label{eq:projected-reduced-order-model}
    \dot{\by} = \pdv{\bpi}{\bx}(\bx)\bF(\bx,\bu).
\end{equation}
While \eqref{eq:projected-reduced-order-model} provides a reduced-order representation of \eqref{eq:full-order-dyn}, it depends on the full-order states and inputs, complicating the design of a reduced-order controller.
We resolve this by defining \emph{idealized} reduced-order dynamics, characterized by locally Lipschitz functions $\bf\,:\,\R^n\rightarrow\R^n$ and $\bg\,:\,\R^n\rightarrow\R^{n\times m}$ that are used to rewrite \eqref{eq:projected-reduced-order-model} as:
\begin{equation}\label{eq:true-reduced-order-model}
    \begin{aligned}
        \dot{\by} = & 
        \bf(\by) + \bg(\by)\bv + \bd,
    \end{aligned}
\end{equation}
where $\bd\coloneqq \pdv{\bpi}{\bx}(\bx)\bF(\bx,\bu)-\bf(\by) -\bg(\by)\bv$ captures the discrepancy between \eqref{eq:projected-reduced-order-model} and the idealized reduced-order dynamics. Ideally, one would choose $\bpi$ and $\bpsi$ such that:
\begin{equation*}
    \pdv{\bpi}{\bx}(\bx)\bF(\bx,\bu) =\bf(\bpi(\bx)) +\bg(\bpi(\bx))\bpsi(\bx),
\end{equation*}
for all $(\bx,\bu)\in\R^N\times\R^M$ so that $\bd\equiv\bzero$, although we will not impose this as a strict requirement. Hereafter, we refer to \eqref{eq:true-reduced-order-model} as a ROM of \eqref{eq:full-order-dyn}. 
To make the preceding developments more concrete, we introduce the following running example.

\begin{example}\label{example:quadrotor-intro}
    Consider a full-order model of a quadrotor from \cite{CosnerICRA24} with state $\bx=(\bp,\bq,\dot{\bp})\in\R^3\times\mathbb{S}^3\times\R^3=\mc{X}$, where $\bp=(x,y,z)\in\R^3$ is the position, $\bq\in\mathbb{S}^3$ is the orientation represented as a unit quaternion, and $\dot{\bp}=(\dot{x},\dot{y},\dot{z})\in\R^3$ is the velocity. The dynamics of the quadrotor are given by:
    \begin{equation}\label{eq:quad-dyn}
    \underbrace{
        \begin{bmatrix}
            \dot{\bp} \\ \dot{\bq} \\ \ddot{\bp}
        \end{bmatrix}}_{\dot{\bx}}
        =
        \underbrace{
        \begin{bmatrix}
            \dot{\bp} \\ \bm{\omega} \\ -\be_{z}g + \frac{1}{m}\bR(\bq)\be_{z}\tau
        \end{bmatrix}}_{\bF(\bx,\bu)},
    \end{equation}
    where the full-order input $\bu=(\bm{\omega},\tau)\in\mathfrak{s}^3\times \R$ is the angular velocity $\bm{\omega}$ and thrust $\tau$.
    Our objective is to control the quadrotor as if it were a two-dimensional single integrator evolving in the plane. To this end, we define
    $\by=\bpi(\bx)\coloneqq (x,y)\in\R^2$, $\bv=\bpsi(\bx)=(\dot{x},\dot{y})\in\R^2$ noting the dynamics of the FOM projected onto the reduced-order space are:
    \begin{equation*}
        \underbrace{
        \begin{bmatrix}
            \dot{x} \\ \dot{y}    
        \end{bmatrix}}_{\dot{\by}}
        =
        \underbrace{
        \begin{bmatrix}
            0 \\ 0
        \end{bmatrix}}_{\bf(\by)}
        +
        \underbrace{
        \begin{bmatrix}
            1 & 0 \\ 0 & 1
        \end{bmatrix}}_{\bg(\by)}
        \underbrace{
        \begin{bmatrix}
            \dot{x} \\ \dot{y}
        \end{bmatrix}}_{\bv},
    \end{equation*}
    which matches the idealized single integrator dynamics (i.e., these projections produce a ROM \eqref{eq:true-reduced-order-model} with $\bd\equiv\bzero$). 
\end{example}

Our main objective in this paper is to design safety-critical controllers $\bkappa\,:\,\R^n\rightarrow\R^m$ for the ROM \eqref{eq:true-reduced-order-model}, which are then refined to ensure safety of the FOM \eqref{eq:full-order-dyn}. To relate the properties of a reduced-order controller $\bkappa$ to \eqref{eq:full-order-dyn}, we define the \emph{reduced-order surface}:
\begin{equation}\label{eq:M0}
    \mc{M} \coloneqq \{\bx\in\R^N\,:\, \bpsi(\bx) - \bkappa(\bpi(\bx)) = \bzero\}.
\end{equation}
Constraining \eqref{eq:full-order-dyn} to $\mc{M}$ along a given trajectory $t\mapsto \bx(t)$ ensures that $\bpsi(\bx(t)) = \bkappa(\bpi(\bx(t)))$ for all $t\geq 0$ so that properties of the reduced-order controller $\bkappa$ may be transferred back to states of the FOM. 
While we will not impose the rather strict assumption that $\mc{M}$ be rendered forward invariant, we will assume that it is possible to drive the FOM to a neighborhood of $\mc{M}$, which is captured by the notion of a simulation function \cite{GirardAutomatica09}, slightly modified to suite the context of this paper.
\begin{definition}\label{def:simulation}
    A continuously differentiable function $V\,:\,\mc{D}\subseteq\R^N\rightarrow\R_{\geq0}$ is said to be a \emph{simulation function} from the ROM \eqref{eq:true-reduced-order-model} to the FOM \eqref{eq:full-order-dyn} with an associated locally Lipschitz interface $\bk\,:\,\R^N\rightarrow\R^M$ if there exists $\rho>0$ such that:
    \begin{equation}\label{eq:lyap1}
    \begin{aligned}
        V(\bx) \geq & \rho \|\bpsi(\bx) - \bkappa(\bpi(\bx))\|^2,
    \end{aligned}
    \end{equation}
    \begin{equation}\label{eq:lyap2}
        \nabla V(\bx)\cdot\bF(\bx,\bk(\bx)) \leq - \lambda V(\bx) + \iota,
    \end{equation}
    for all $\bx\in\mc{D}$, where $\lambda>0$ and $\iota\geq0$ satisfy $\beta> \iota/\lambda$ with:
    \begin{equation*}
        \Omega_{\beta}\coloneqq \{\bx\in\mc{D}\,:\,V(\bx) \leq \beta\},
    \end{equation*}
    the largest sublevel set of $V$ contained within $\mc{D}$.
\end{definition}

\begin{lemma}\label{lemma:simulation}
    Let $V\,:\,\mc{D}\rightarrow\R_{\geq0}$ be a simulation function from \eqref{eq:true-reduced-order-model} to \eqref{eq:full-order-dyn} with an associated interface $\bk\,:\,\R^N\rightarrow\R^M$. Then, for any initial condition $\bx_0\in\Omega_{\beta}$, trajectories $t\mapsto\bx(t)$ of closed-loop full-order system $\dot{\bx}=\bF(\bx,\bk(\bx))$ satisfy:
    \begin{equation}
        \|\bpsi(\bx(t)) - \bkappa(\bpi(\bx(t)))\|^2 \leq \frac{V(\bx_0) }{\rho}e^{-\lambda t}+ \frac{\iota}{\lambda \rho},\; \forall t\in I(\bx_0).
    \end{equation}
\end{lemma}

The proof of Lemma \ref{lemma:simulation} is a direct consequence of the Comparison Lemma \cite[Lemma 3.4]{Khalil}. The following result will be useful when discussing set invariance properties. 

\begin{lemma}\label{lemma:beta-regular-value}
    If $V\,:\,\R^N\rightarrow\R_{\geq0}$ is a simulation function from \eqref{eq:true-reduced-order-model} to \eqref{eq:full-order-dyn} with interface $\bk$, then $\beta$ is a regular value of $V$.
\end{lemma}
\begin{proof}
    For the sake of contradiction assume $\beta$ is not a regular value of $V$. Then, when $V(\bx)=\beta$ we have $\nabla V(\bx)=\bzero$ and $\nabla V(\bx)\cdot\bF(\bx,\bk(\bx))=0$, which, by \eqref{eq:lyap2}, implies that:
    \begin{equation*}
         0 \leq - \lambda V(\bx) + \iota = -\lambda\beta + \iota < -\lambda\frac{\iota}{\lambda} + \iota = 0,
    \end{equation*}
    where the final inequality follows from the fact that $\beta > \iota/\lambda$. As we cannot have $0<0$, the above contradicts the initial claim that $\beta$ is not a regular value of $V$. Hence, by contradiction, $\beta$ must be a regular value of $V$.
\end{proof}

The existence of a simulation function and associated interface allows for refining inputs for a ROM to those for the FOM. When implemented on the FOM, this interface ensures that states within the control projection map $\bpsi(\bx)$ remain close to the inputs generated by the ROM controller $\bkappa(\bpi(\bx))$, allowing to relate properties of the reduced-order controller $\bkappa$ to the full-order state $\bx$.

\section{Safety-Critical Control with ROMs}\label{sec:safety}
We now discuss how the preceding framework may be used to design safety-critical controllers for complex systems based on their corresponding ROMs. To this end, consider a state constraint set for a FOM:
\begin{align}\label{eq:C}
    \mc{C} \coloneqq & \{\bx\in\R^N\,:\,h(\bpi(\bx))\geq 0\},
\end{align}
where $h\,:\,\R^n\rightarrow\R$ is continuously differentiable. While $\mc{C}$ exists in the full-order space, it only depends on the states related to the ROM. The derivative of $h$ along the full-order dynamics is given by:
\begin{equation}\label{eq:h-dot}
\begin{aligned}
    \dot{h} = & \nabla h(\bpi(\bx))\cdot \pdv{\bpi}{\bx}(\bx)\bF(\bx,\bu) \\
    = & L_{\bf}h(\by) + L_{\bg}h(\by)\bv + \nabla h(\by)\cdot \bd.
\end{aligned}
\end{equation}
Similar to backstepping \cite{Krstic}, we view $\bv=\bpsi(\bx)$, the input to the ROM, as a ``virtual" control input and design a controller $\bkappa\,:\,\R^n\rightarrow\R^m$ for the ROM that would enforce forward invariance of $\mc{C}$, provided its dynamics were directly controllable. Since the inputs to the ROM are not the same as those of the FOM, we then relate the inputs of the ROM $\bv=\bkappa(\by)$ to those of the FOM \eqref{eq:full-order-dyn} via a simulation function $V$ and interface $\bk$. To design the ROM controller, we leverage CBFs: suppose there exists a locally Lipschitz controller $\bkappa\,:\,\R^n\rightarrow\R^m$ satisfying:
\begin{equation}\label{eq:issf-rom}
    \begin{aligned}
        L_{\bf}h(\by) + L_{\bg}h(\by)\bkappa(\by) > & - \alpha h(\by) + \frac{1}{\varepsilon}\|L_{\bg}h(\by)\|^2 \\ &  + \frac{1}{\sigma} \|\nabla h(\by)\|^2
    \end{aligned}
\end{equation}
for all $\by\in\R^n$, where $\alpha,\varepsilon,\sigma>0$. The above condition effectively requires $h$ to be an \emph{input-to-state safe} CBF \cite{AnilTCST23} for the ROM \eqref{eq:true-reduced-order-model}, where the last two terms in \eqref{eq:issf-rom} are included to compensate for transient tracking errors $\|\bpsi(\bx) - \bkappa(\bpi(\bx))\|$ of the interface $\bk$ from Def. \ref{def:simulation} and to compensate for $\bd$ from \eqref{eq:true-reduced-order-model}. We now combine $h$ with a simulation function $V$ to form the barrier function candidate for the FOM:
\begin{equation}\label{eq:B}
    B(\bx) = h(\bpi(\bx)) - \frac{1}{\mu} V(\bx),
\end{equation}
where $\mu\in\R_{>0}$,
which is used to define a candidate safe set:
\begin{equation}\label{eq:S2}
\begin{aligned}
    \mc{S} \coloneqq & \{\bx\in\R^N\,:\,B(\bx) \geq 0\},
\end{aligned}
\end{equation}
for the FOM \eqref{eq:full-order-dyn}. Since $V(\bx)\geq0$, we have $B(\bx)\geq 0 \implies h(\bpi(\bx))\geq0$ so that rendering $\mc{S}$ forward invariant leads to satisfaction of the state constraint in \eqref{eq:C}. 
Before proceeding, we note that with $\bu=\bk(\bx)$,  $\bd$ may be written as:
\begin{equation}\label{eq:d-X}
\begin{aligned}
    \bd(\bx) = & \pdv{\bpi}{\bx}(\bx)\bF(\bx,\bk(\bx))
    - \bf(\bpi(\bx)) - \bg(\bpi(\bx))\bpsi(\bx).
\end{aligned}
\end{equation}
The following theorem constitutes the main result of this paper and illustrates that when there exists a simulation function from the ROM to the FOM, then one may refine reduced-order controllers to ensure safety of the FOM.

\begin{theorem}\label{theorem:main}
    Consider the FOM \eqref{eq:full-order-dyn}, the ROM \eqref{eq:true-reduced-order-model}, the set $\mc{C}\subset\R^n$ as in \eqref{eq:C}, and suppose there exists a simulation function $V\,:\,\mc{D}\rightarrow\R_{\geq0}$ from \eqref{eq:true-reduced-order-model} to \eqref{eq:full-order-dyn} with an associated locally Lipschitz interface $\bk\,:\,\R^N\rightarrow\R^M$. Define:
    \begin{equation}
        \mc{S}_{\delta} \coloneqq \left\{\bx\in\R^N\,:\,B(\bx) + \frac{1}{\alpha}\left(\frac{\sigma}{4}\delta^2 + \frac{\iota}{\mu} \right) \geq 0 \right\},
    \end{equation}
    where $B$ is defined as in \eqref{eq:B} and $\delta\coloneqq \sup_{\bx\in\Omega_{\beta}}\bd(\bx)$ with $\bd$ as in \eqref{eq:d-X}. Provided that:
    \begin{equation}\label{eq:gain-conditions}
        \lambda \geq \alpha + \frac{\varepsilon\mu}{4\rho},
    \end{equation}
    then $\mc{W}\coloneqq \mc{S}_{\delta}\cap\Omega_{\beta}$ is forward invariant for the closed-loop full-order system \eqref{eq:full-order-dyn} with $\bu=\bk(\bx)$. 
\end{theorem}
Note that if $\bd\equiv\bzero$ and $\iota=0$ then $\mc{S}\cap\Omega_{\beta}$ is forward invariant and that for a fixed interface $\bk$ satisfying \eqref{eq:lyap2}, it is always possible to satisfy \eqref{eq:gain-conditions} by decreasing both $\alpha$ and $\varepsilon$. Theorem \ref{theorem:main} states that the intersection of $\Omega_{\beta}$ and an \emph{inflated} safe set $\mc{S}_{\delta}\supseteq\mc{S}$ is rendered forward invariant, where the inflation is proportional to $\delta$, the bound on $\bd$ from \eqref{eq:d-X}, and $\iota$, which characterizes the bound on $\|\bpsi(\bx) - \bkappa(\bpi(\bx))\|$ in Lemma \ref{lemma:simulation}. The size of this inflation can be shrunk by decreasing $\sigma$, which is decoupled from \eqref{eq:gain-conditions}, and by increasing $\alpha$ and $\mu$, which are coupled to \eqref{eq:gain-conditions}. Importantly, while there may exist points along a given trajectory such that $B(\bx(t))\leq0$, this does not necessarily imply violation of the state constraint from \eqref{eq:C} since $h(\bpi(\bx))\geq B(\bx)$.

\begin{proof}
    Define:
    \begin{equation}\label{eq:B-delta}
        B_{\delta}(\bx) \coloneqq B(\bx) + \frac{1}{\alpha}\left(\frac{\sigma}{4}\delta^2 + \frac{\iota}{\mu} \right),
    \end{equation}
    noting that $\mc{S}_{\delta}$ from \eqref{eq:S2} is the zero superlevel set of $B_{\delta}$.    Computing the derivative of $B_{\delta}$ along the closed-loop full-order dynamics yields:
    \begin{equation*}
        \begin{aligned}
            \dot{B}_{\delta}(\bx) 
            = & L_{\bf}h(\by) +L_{\bg}h(\by)\bv + \nabla h(\by)\cdot \bd - \frac{1}{\mu}\dot{V}(\bx) \\
            = & L_{\bf}h(\by) +L_{\bg}h(\by)\bkappa(\by) + \nabla h(\by)\cdot \bd \\ &+ L_{\bg}h(\by)(\bv-\bkappa(\by)) - \frac{1}{\mu} \nabla V(\bx)\cdot\bF(\bx,\bk(\bx)), \\
        \end{aligned}
    \end{equation*}
    where the first equality follows from \eqref{eq:B} and \eqref{eq:h-dot} and the second from adding zero. For notational brevity, $\by$ and $\bv$ are viewed as functions of $\bx$ via $\by=\bpi(\bx)$ and $\bv=\bpsi(\bx)$. Lower bounding the above on $\mc{W}$ using \eqref{eq:issf-rom} and \eqref{eq:lyap2} yields:
    \begin{equation*}
        \begin{aligned}
            \dot{B}_{\delta}(\bx) > & -\alpha h(\by) + \frac{1}{\varepsilon}\|L_{\bg}h(\by)\|^2 + \frac{1}{\sigma}\|\nabla h(\by)\|^2 \\ &- \|\nabla h(\by)\|\|\bd\| - \|L_{\bg}h(\by)\|\|\bv - \bkappa(\by)\| \\ 
            & + \frac{\lambda}{\mu}V(\bx) - \frac{\iota}{\mu}.
        \end{aligned}
    \end{equation*}
    By completing squares and lower bounding, we obtain:
    \begin{equation*}
        \begin{aligned}
        \dot{B}_{\delta}(\bx) > & -\alpha h(\by) - \frac{\varepsilon}{4}\|\bv - \bkappa(\by)\|^2 - \frac{\sigma}{4}\delta^2 + \frac{\lambda}{\mu}V(\bx) - \frac{\iota}{\mu} \\
        = & -\alpha B_{\delta}(\bx) + \frac{1}{\mu}\left(\lambda - \alpha \right)V(\bx) - \frac{\varepsilon}{4}\|\bv - \bkappa(\by)\|^2,
        \end{aligned}
    \end{equation*}
    where the equality follows from \eqref{eq:B} and \eqref{eq:B-delta}. Using \eqref{eq:lyap1}:
    \begin{equation*}
        \begin{aligned}
            \dot{B}_{\delta}(\bx) > -\alpha B_{\delta}(\bx) + \frac{1}{\mu}\left(\lambda - \alpha - \frac{\varepsilon\mu}{4 \rho} \right)V(\bx). \\
        \end{aligned}
    \end{equation*}
    Hence, provided \eqref{eq:gain-conditions} holds, we have:
    \begin{equation}\label{eq:B-delta-dot}
        \begin{aligned}
            \dot{B}_{\delta}(\bx) > -\alpha B_{\delta}(\bx), \quad \mc{\forall \bx\in\mc{W}}.
        \end{aligned}
    \end{equation}
    Now define $B_{\beta}(\bx)\coloneqq \beta - V(\bx)$.
    This function satisfies:
    \begin{equation*}
        \begin{aligned}
        \dot{B}_{\beta}(\bx) = & -\nabla V(\bx)\cdot \bF(\bx,\bk(\bx)) \\ \geq & \lambda V(\bx) - \iota
        = -\lambda B_{\beta}(\bx) + \lambda\beta - \iota > -\lambda B_{\beta}(\bx),
        \end{aligned}
    \end{equation*}
    $\forall \bx\in\mc{W}$, where the first inequality follows from \eqref{eq:lyap2} and the second from $\beta> \iota/\lambda$.
    With $B_{\delta}$ and $B_{\beta}$ we have that:
    \begin{equation*}
        \mc{W} = \{\bx\in\R^N\,:\,B_{\delta}(\bx)\geq 0 \wedge B_{\beta}(\bx) \geq 0\}.
    \end{equation*}
    Since \eqref{eq:B-delta-dot} holds with strict inequality, we have $\nabla B_{\delta}(\bx)\neq\bzero$ whenever $B_{\delta}(\bx)=0$ implying that zero is a regular value of $B_{\delta}$ (this can be shown using a similar argument to Lemma \ref{lemma:beta-regular-value}). Furthermore, since $V$ is a simulation function $\beta$ is a regular value of $V$ by Lemma \ref{lemma:beta-regular-value}, which implies that $0$ is a regular value of $B_{\beta}$.
    Let $\act(\bx)\coloneqq \{i\in\{\beta,\delta\}\,:\,B_{i}(\bx) = 0 \}$ denote the set of active constraints of $\mc{W}$ and note that since zero is a regular value of $B_{\beta}$ and $B_{\delta}$, we have \cite[Ch. 4]{Blanchini}:
    \begin{equation*}
        \mc{T}_{\mc{W}}(\bx) = \{\bz\in\R^N\,:\,\nabla B_{i}(\bx)\cdot \bz \geq 0,\quad \forall i\in\act(\bx)\}.
    \end{equation*}
    When $B_{i}(\bx) = 0$ we have:
    \begin{equation*}
        \dot{B}_{i}(\bx) = \nabla B_{i}(\bx)\cdot \bF(\bx,\bk(\bx)) > 0,
    \end{equation*}
    ${\forall i\in\act(\bx)}$.
    Hence, for each ${\bx\in\partial\mc{W}}$ we have $\bF(\bx,\bk(\bx))\in\mc{T}_{\mc{W}}(\bx)$, which, by Theorem \ref{theorem:nagumo}, implies the forward invariance of $\mc{W}$, as desired. 
\end{proof} 

\begin{remark}
Constructing a ROM, simulation function, and corresponding interface is a challenging problem in general; however, various methods exist for certain classes of systems. For linear FOMs, \cite{PappasTAC00,GirardAutomatica09} provide systematic methods for constructing linear ROMs and corresponding simulation functions with linear interfaces \cite{GirardAutomatica09}. When the FOM \eqref{eq:full-order-dyn} is control affine, \cite{PappasTAC02} provides a method to construct control affine ROMs, and methods based on sum-of-squares programming \cite{ArcakARC22}, backstepping \cite{Schweidel}, feedback linearization \cite{FuACC13}, and differential flatness \cite{GirardECC13} may be used to construct simulation functions and corresponding interfaces. As demonstrated in the following section, the recently developed notion of zero dynamics policies \cite{ComptonCDC24} may also be leveraged to construct simulation functions and interfaces for highly underactuated systems. Any of these methods may be integrated with our ROM framework, although we emphasize that explicit expressions for both $V$ and $\bk$ are not necessary for \emph{implementation} of our approach (bounds on $V$ are required for \emph{verification} of the conditions in Theorem \ref{theorem:main}). When such conditions cannot be easily verified (e.g., when $V$ is unknown and $\bk$ is a black-box component in an existing control architecture), the conditions of Theorem \ref{theorem:main} may still be satisfied by initializing $\alpha$ and $\varepsilon$ very small to ensure safety and then increasing such parameters until adequate performance is achieved.
\end{remark}

\begin{example}\label{example:quadrotor-sims}
    Continuing Example \ref{example:quadrotor-intro}, suppose our objective is to drive the quadrotor to a goal while avoiding a cylindrical obstacle centered at $(x,y)=(x_o,y_o)\in\R^2$ with a radius of $r_o\in\R_{>0}$. This requirement leads to the state constraint $h(\bpi(\bx))\coloneqq\|\bpi(\bx) - (x_o,y_o)\|^2 - r_o^2$ for the FOM, which defines a state constraint set $\mc{C}$ as in \eqref{eq:C}. This state constraint is a valid input-to-state safe CBF \cite{AnilTCST23} for the single integrator ROM, which may be used to synthesize a controller $\bkappa$ satisfying \eqref{eq:issf-rom} using a quadratic program \cite{AmesTAC17} or a smooth safety filter \cite{CohenLCSS23}. This reduced-order controller is then refined to produce inputs for the FOM using an off-the-shelf tracking controller (e.g., that in \cite{LeeCDC10}) as an interface between the ROM and FOM. 
    Example trajectories of the quadrotor using this approach are illustrated in Fig. \ref{fig:quadrotor}, where the left plot displays the position of the quadrotor, which attempts to track velocities from the ROM for different choices of $\alpha$, and the right plot illustrates the difference between the reduced-order controller and the velocity of the quadrotor. As indicated by condition~\eqref{eq:gain-conditions} of Theorem \ref{theorem:main}, picking $\alpha$ too large may lead to safety violations (blue curve), whereas decreasing $\alpha$ allows such conditions to be satisfied for a fixed tracking controller and ensures safety (red and green curves). 
\end{example}

\begin{figure}
    \centering
    \includegraphics{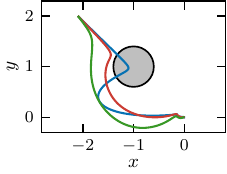}
    \hfill
    \includegraphics{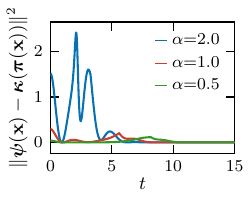}
    \vspace{-0.4cm}
    \caption{(\textbf{Left}) Evolution of the quadrotor's position for different choices of $\alpha$ in \eqref{eq:issf-rom}. (\textbf{Right}) Difference between quadrotor's planar velocity $\bv=\bpsi(\bx)$ and desired velocity $\bkappa(\bpi(\bx))$ generated by a single integrator. The velocity error remains bounded per Lemma \ref{lemma:simulation}, although the error is too large when $\alpha=2$, leading to safety violations. For each of these simulations we took $\varepsilon=20$ and omitted the $\sigma$ term since $\nabla h = L_{\bg}h$. Varying $\varepsilon$ had minimal effect on trajectories.}
    \label{fig:quadrotor}
\end{figure}

\section{Hardware Demonstrations}\label{sec:hardware}
We now illustrate our developments through their application\footnote{\texttt{https://vimeo.com/1010060590?share=copy}} to safety-critical control of ARCHER, a 3D hopping robot \cite{Ambrose,NoelICRA23,NoelIROS24}. 
The state of ARCHER is described by $\bx=(\bp,\bq,\dot{\bp},\bomega)\in\R^3\times\mathbb{S}^3\times\R^3\times\mathfrak{s}^3\eqqcolon\mc{X}$, where $\bp$ denotes the position of the robot, $\bq$ a unit quaternion representing its orientation, and $\bomega$ its angular rates. 
Overall, ARCHER is a high-dimensional, hybrid, underactuated system, which prohibits the use of traditional CBF synthesis methods. 
To overcome this, we leverage a ROM paradigm by designing a safety-critical controller for an abstracted version of ARCHER, the inputs of which are then refined for the full-order system using an existing high-performance interface \cite{ComptonCDC24,NoelIROS24}. To construct a ROM of ARCHER, we define our state and control projection maps as $\bpi(\bx)\coloneqq (x,y)\in\R^2$ and $\bpsi(\bx)\coloneqq (\dot{x},\dot{y})\in\R^2$, where $\bp=(x,y,z)$, so that the idealized reduced-order dynamics are:
\begin{equation*}
        \underbrace{
        \begin{bmatrix}
            \dot{x} \\ \dot{y}    
        \end{bmatrix}}_{\dot{\by}}
        =
        \underbrace{
        \begin{bmatrix}
            0 \\ 0
        \end{bmatrix}}_{\bf(\by)}
        +
        \underbrace{
        \begin{bmatrix}
            1 & 0 \\ 0 & 1
        \end{bmatrix}}_{\bg(\by)}
        \underbrace{
        \begin{bmatrix}
            \dot{x} \\ \dot{y}
        \end{bmatrix}}_{\bv}
        +
        \bd,
\end{equation*}
which corresponds to a disturbed single integrator evolving in the $(x,y)$ plane. The inputs of this ROM are related back to those of ARCHER using a zero dynamics policy \cite{ComptonCDC24,NoelIROS24} as an interface while this policy's associated Lyapunov function is used as a simulation function to certify adequate tracking of the ROM. Similar to Example \ref{example:quadrotor-sims}, our objective is to design a controller that avoids a collection of planar obstacles, each of which is captured by a safe set $\mc{C}_i$ and CBF $h_i$ as in Example \ref{example:quadrotor-sims}. We synthesize a safety-critical controller for this ROM using the safety filter:
\begin{equation*}
\begin{aligned}
    \bm{\kappa}(\by) \coloneqq && \argmin_{\bv\in\R^2} & \quad \tfrac{1}{2}\|\bv - \bm{\kappa}_{\rm{d}}(\by)\|^2 \\ 
    && \mathrm{s.t.} & \quad \nabla h(\by)\cdot \bv \geq - \alpha h(\by) + \frac{1}{\varepsilon}\|\nabla h(\by)\|^2,
\end{aligned}
\end{equation*}
where $\bm{\kappa}_{\rm{d}}\,:\,\R^2\rightarrow\R^2$ is a desired reduced-order input, and $h$ combines each individual CBF into a single CBF \cite{TamasLCSS23}. For our demonstration, $\bm{\kappa}_{\rm{d}}$ corresponds to desired velocity commands given via joystick that attempt to drive the hopper from one location to another without accounting for safety. This unsafe velocity is then passed to the above safety filter, which outputs a safe velocity command for ARCHER that is tracked by the corresponding interface $\bk$. The results of applying this interface to ARCHER are illustrated in Fig. \ref{fig:hopper-position} and Fig. \ref{fig:hopper-velocity}. As shown in Fig. \ref{fig:hopper-position}, the resulting full-order trajectory is safe -- it avoids obstacles at all times as indicated by the positivity of $h(\bm{\pi}(\bx(t)))$ for all $t\geq0$. The safe velocity commands generated by the ROM and the velocities achieved by the full-order system are illustrated in Fig. \ref{fig:hopper-velocity}, where the difference between the true velocities $\bpsi(\bx)$ and desired velocities $\bm{\kappa}(\bm{\pi}(\bx))$ are bounded per Lemma \ref{lemma:simulation}. The fact that the full-order model cannot track the commanded velocities exactly is compensated for by using a relatively small $\alpha=0.4$ to satisfy the conditions of Theorem \ref{theorem:main}.

\begin{figure}
    \centering
    \includegraphics{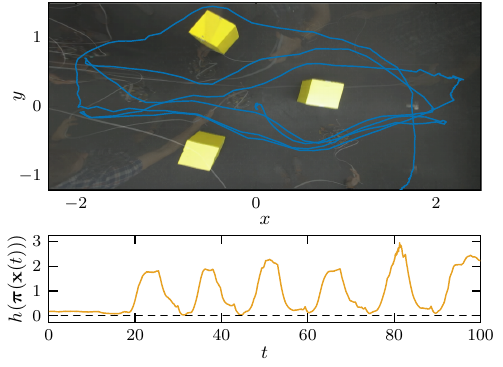}
    \vspace{-0.4cm}
    \caption{(\textbf{Top}) Evolution of ARCHER's position where the 
    yellow cubes
    denote the obstacles. (\textbf{Bottom}) Evolution of the ROM's CBF $h$ along the trajectory of the full-order system, which remains positive for all time.}
    \label{fig:hopper-position}
\end{figure}

\begin{figure}
    \centering
    \includegraphics[]{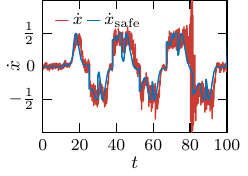}
    \hfill
    \includegraphics[]{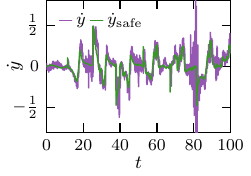}
    \vspace{-0.8cm}
    \caption{Commanded velocities output by the reduced-order safety filter $\bm{\kappa}(\bm{\pi}(\bx))=(\dot{x}_{\rm{safe}}, \dot{y}_{\rm{safe}})$ and the velocities of the full-order system $(\dot{x},\dot{y})$.}
    \label{fig:hopper-velocity}
\end{figure}

\section{Conclusions}\label{sec:conclusions}
We presented a framework for safety-critical control with CBFs and ROMs in which the relation between full and reduced-order models is characterized using simulation functions. This paradigm facilitates the use of highly simplified models for safety-critical control design while still providing safety guarantees for the original full-order model. Our theoretical developments were illustrated in both simulation and via hardware demonstrations on a hopping robot.

\bibliographystyle{ieeetr}
\bibliography{biblio}

\begin{thebibliography}{10}

\bibitem{AmesTAC17}
A.~D. Ames, X.~Xu, J.~W. Grizzle, and P.~Tabuada, ``Control barrier function based quadratic programs for safety critical systems,'' {\em IEEE Trans. Autom. Control}, vol.~62, no.~8, pp.~3861--3876, 2017.

\bibitem{CohenARC24}
M.~H. Cohen, T.~G. Molnar, and A.~D. Ames, ``Safety-critical control for autonomous systems: Control barrier functions via reduced order models,'' {\em Annual Reviews in Control}, vol.~57, p.~100947, 2024.

\bibitem{TamasRAL22}
T.~G. Molnar, R.~K. Cosner, A.~W. Singletary, W.~Ubellacker, and A.~D. Ames, ``Model-free safety-critical control for robotic systems,'' {\em IEEE Robot. Aut. Lett.}, vol.~7, no.~2, pp.~944--951, 2022.

\bibitem{TamasACC23}
T.~G. Molnar and A.~D. Ames, ``Safety-critical control with bounded inputs via reduced order models,'' in {\em Proc. Amer. Control Conf.}, pp.~1414--1421, 2023.

\bibitem{Krstic}
M.~Krsti\'{c}, I.~Kanellakopoulus, and P.~Kokotovi\'{c}, {\em Nonlinear and Adaptive Control Design}.
\newblock Wiley, 1995.

\bibitem{TomlinTAC21}
M.~Chen, S.~L. Herbert, H.~Hu, Y.~Pu, J.~F. Fisac, S.~Bansal, S.~Han, and C.~J. Tomlin, ``Fastrack: A modular framework for real-time motion planning and guaranteed safe tracking,'' {\em IEEE Trans. Autom. Control}, vol.~66, no.~12, pp.~5861--5876, 2021.

\bibitem{ArcakARC22}
K.~S. Schweidel, H.~Yin, S.~W. Smith, and M.~Arcak, ``Safe-by-design planner–tracker synthesis with a hierarchy of system models,'' {\em Annual Reviews in Control}, vol.~53, pp.~138--146, 2022.

\bibitem{BendersArXiv24}
D.~Benders, J.~K\"{o}hlher, T.~Niesten, R.~Bab\u{u}ska, J.~Alonso-Mora, and L.~Ferranti, ``Embedded hierarchical mpc for autonomous navigation,'' {\em arXiv preprint arXiv:2406.11506}, 2024.

\bibitem{SinghIJRR23}
S.~Singh, B.~Landry, A.~Majumdar, J.~J. Slotine, and M.~Pavone, ``Robust feedback motion planning via contraction theory,'' {\em The International Journal of Robotics Research}, vol.~42, no.~9, pp.~655--688, 2023.

\bibitem{Tabuada}
P.~Tabuada, {\em Verification and control of hybrid systems: a symbolic approach}.
\newblock Spring Science \& Business Media, 2009.

\bibitem{Belta}
C.~Belta, B.~Yordanov, and E.~A. Gol, {\em Formal methods for discrete-time dynamical systems}.
\newblock Springer, 2017.

\bibitem{PappasTAC00}
G.~J. Pappas, G.~Lafferriere, and S.~Sastry, ``Hierarchically consistent control systems,'' {\em IEEE Trans. Autom. Control}, vol.~45, no.~6, pp.~1144--1160, 2000.

\bibitem{PappasTAC02}
G.~J. Pappas and S.~Simi\'{c}, ``Consistent abstractions of affine control systems,'' {\em IEEE Trans. Autom. Control}, vol.~47, no.~5, pp.~745--756, 2002.

\bibitem{GirardAutomatica09}
A.~Girard and G.~J. Pappas, ``Hierarchical control system design using approximate simulation,'' {\em Automatica}, vol.~45, pp.~566--571, 2009.

\bibitem{LeeCDC10}
T.~Lee, M.~Leoky, and N.~H. McClamroch, ``Geometric tracking control of a quadrotor {UAV} on {SE(3)},'' in {\em Proc. Conf. Decis. Control}, pp.~5420--5425, 2010.

\bibitem{MellingerICRA2011}
D.~Mellinger and V.~Kumar, ``Minimum snap trajectory generation and control for quadrotors,'' in {\em Proc. Int. Conf. Robot. and Autom.}, pp.~2520--2525, 2011.

\bibitem{WensingTRO24}
P.~M. Wensing, M.~Posa, Y.~Hu, A.~Escande, N.~Mansard, and A.~D. Prete, ``Optimization-based control for dynamic legged robots,'' {\em IEEE Trans. Robot}, vol.~40, 2024.

\bibitem{HutterScience19}
J.~Hwangbo, J.~Lee, A.~Dosovitskiy, D.~Bellicoso, V.~Tsounis, V.~Koltun, and M.~Hutter, ``Learning agile and dynamic motor skills for legged robots,'' {\em Science Robotics}, vol.~4, no.~26, 2019.

\bibitem{Ambrose}
E.~Ambrose, {\em Creating ARCHER: A 3D Hopping Robot with Flywheels for Attitude Control}.
\newblock PhD thesis, California Institute of Technology, 2022.

\bibitem{NoelICRA23}
N.~Csomay-Shanklin, V.~D. Dorobantu, and A.~D. Ames, ``Nonlinear model predictive control of a {3D} hopping robot: Leveraging {Lie} group integrators for dynamically stable behaviors,'' in {\em Proc. Int. Conf. Robot. and Autom.}, pp.~12106--12112, 2023.

\bibitem{NoelIROS24}
N.~Csomay-Shanklin, W.~Compton, I.~D.~J. Rodriguez, E.~R. Ambrose, Y.~Yue, and A.~D. Ames, ``Robust agility via learned zero dynamics policies,'' {\em arXiv preprint arXiv:2409.06125}, 2024.

\bibitem{Blanchini}
F.~Blanchini and S.~Miani, {\em Set-theoretic methods in control}.
\newblock Springer, 2008.

\bibitem{AbrahamMarsdenRatiu}
R.~Abraham, J.~E. Marsden, and T.~Ratiu, {\em Manifolds, tensor analysis, and applications}.
\newblock Addison-Wesley, 1983.

\bibitem{CosnerICRA24}
R.~K. Cosner, I.~Sadalski, J.~K. Woo, P.~Culbertson, and A.~D. Ames, ``Generative modeling of residuals for real-time risk-sensitive safety with discrete-time control barrier functions,'' in {\em Proc. Int. Conf. Robot. and Autom.}, pp.~9960--9967, 2024.

\bibitem{Khalil}
H.~K. Khalil, {\em Nonlinear Systems}.
\newblock Prentice Hall, 3~ed., 2002.

\bibitem{AnilTCST23}
A.~Alan, A.~J. Taylor, C.~R. He, A.~D. Ames, and G.~Orosz, ``Control barrier functions and input-to-state safety with application to automated vehicles,'' {\em IEEE Trans. Contr. Syst. Tech.}, vol.~31, no.~6, pp.~2744--2759, 2023.

\bibitem{Schweidel}
K.~S. Schweidel, {\em Robust Hierarchical Control with Connected Layers}.
\newblock PhD thesis, University of California, Berkeley, 2023.

\bibitem{FuACC13}
J.~Fu, S.~Shah, and H.~G. Tanner, ``Hierarchical control via approximate simulation and feedback linearization,'' in {\em Proc. Amer. Control Conf.}, pp.~1816--1821, 2013.

\bibitem{GirardECC13}
A.~Colombo and A.~Girard, ``An approximate abstraction approach to safety control of differentially flat systems,'' in {\em Proc. Eur. Control Conf.}, pp.~4226--4231, 2013.

\bibitem{ComptonCDC24}
W.~Compton, I.~D.~J. Rodriguez, N.~Csomay-Shanklin, Y.~Yue, and A.~D. Ames, ``Constructive nonlinear control of underactuated systems via zero dynamics policies,'' {\em arXiv preprint arXiv:2408.14749}, 2024.

\bibitem{CohenLCSS23}
M.~H. Cohen, P.~Ong, G.~Bahati, and A.~D. Ames, ``Characterizing smooth safety filters via the implicit function theorem,'' {\em IEEE Contr. Syst. Lett.}, vol.~7, pp.~3890--3895, 2023.

\bibitem{TamasLCSS23}
T.~G. Molnar and A.~D. Ames, ``Composing control barrier functions for complex safety specifications,'' {\em IEEE Contr. Syst. Lett.}, vol.~7, pp.~3615--3620, 2023.

\end{thebibliography}

\end{document}